\documentclass{article}

\PassOptionsToPackage{hyphens}{url}
\usepackage{hyperref}
\usepackage{amsmath}
\usepackage{amsfonts}
\usepackage{amssymb}
\usepackage{amsthm}
\usepackage{tikz}
\usepackage{caption}
\usetikzlibrary{shapes,arrows,positioning,calc,fit}
\usepackage{algpseudocode,algorithm,algorithmicx} 
\usepackage[numbers]{natbib}
\newtheorem{theorem}{Theorem}
\newtheorem{proposition}[theorem]{Proposition}
\newtheorem{lemma}[theorem]{Lemma}
\newtheorem{corollary}[theorem]{Corollary}
\newtheorem{example}[theorem]{Example}
\tikzstyle{ledger} = [rectangle, draw, 
    text width=2.5em, text centered, rounded corners, minimum height=2.5em]
\tikzstyle{wideledger} = [rectangle, draw, 
    text width=4em, text centered, rounded corners, minimum height=2.5em]
\tikzstyle{line} = [draw, -latex']
\tikzstyle{peer} = [circle, draw, fill=white]
\pgfdeclarelayer{background}
\pgfdeclarelayer{foreground}
\pgfsetlayers{background,main,foreground}
\algdef{SE}[RECEIVE]{Receive}{EndReceive}[1]{\textbf{receive}\ #1\ \algorithmicdo}{\algorithmicend\ \textbf{receive}}%

\begin{document}
\title{Analysis of the XRP Ledger Consensus Protocol}
\author{Brad Chase \and Ethan MacBrough}
\date{Ripple Research \\
	\texttt{\{bchase,emacbrough\}@ripple.com} \\ \vspace*{0.5em}
	\today}
\maketitle
\begin{abstract}
    The XRP Ledger Consensus Protocol is a previously developed consensus protocol powering the XRP Ledger. It is a low-latency Byzantine agreement protocol, capable of reaching consensus without full agreement on which nodes are members of the network. We present a detailed explanation of the algorithm and derive conditions for its safety and liveness. 
\end{abstract}

\section{Introduction}  
The XRP Ledger is a distributed payment system enabling users to transfer value seamlessly around the world. Operating within a distributed peer-to-peer network, the XRP Ledger faces the same challenges as other digital currencies in preventing double-spending of funds and ensuring network-wide consensus on the state of user accounts and balances. First proposed and then implemented by Schwartz et al. \citep{schwartz2014ripple}, the algorithm underlying XRP solves these problems using a Byzantine fault tolerant agreement protocol over collectively trusted subnetworks, hereby referred to as the XRP Ledger Consensus Protocol, or XRP LCP for short.

Abstractly, the XRP Ledger network is a replicated state machine \citep{Schneider1993RMU}. The replicated state is the ledger maintained by each node in the network and state transitions correspond to transactions submitted by clients of the network. Once nodes agree on sets of transactions to apply to the state, a transaction processing protocol specifies deterministic rules for ordering transactions within each set and how to apply transactions to generate the new ledger state. Thus, the role of XRP LCP is only to make the network reach agreement on sets of transactions, not on the content or outcome of those transactions. As long as nodes agree on a transaction set, the transaction processing protocol guarantees that every node generates a consistent ledger. As a Byzantine fault tolerant protocol, XRP LCP must operate even in the presence of faulty or malicious participants.

Byzantine fault tolerant consensus protocols have a rich history, but most require preexisting agreement on the protocol participants \citep{Pease1980RAP,Castro1999}. The distinguishing characteristic of XRP LCP is that it guarantees consistency with only partial agreement on who participates, allowing a decentralized open network. Compared to other decentralized consensus algorithms like proof-of-work \citep{nakamoto2012bitcoin} or proof-of-stake \citep{2017arXiv171009437B}, XRP LCP has since its inception provided lower transaction latency and higher throughput for its users. However, without uniform agreement on the network participants, users still need a way to determine whether their choice of network peers will lead to a consistent network state. In this setting, each user individually defines a \textbf{unique node list} or \textbf{UNL}, which is the set of nodes whose messages it will listen to when making decisions about the network state. It is the intersection of any pair of correct nodes' UNLs that determines network safety. As described in the original whitepaper \citep{schwartz2014ripple}, the minimum overlap requirement was originally believed to be roughly $20\%$ of the UNL. An independent analysis later suggested the correct bound was instead roughly $>40\%$ \citep{Armknecht2015}. 

Given this confusion, our goal in this work is to give a clear and detailed explanation of XRP LCP and derive the necessary conditions on UNL overlap for consistency and liveness. We will not discuss the transactional semantics of XRP's ledger or XRP's benefits as a digital currency, but instead view the algorithm as a general consensus protocol. We re-evaluate the two prior overlap results and provide a single corrected bound which is partway between the bounds of \citep{schwartz2014ripple} and \citep{Armknecht2015}. We also show that under a more general fault model which was not considered in the original whitepaper but is canonically used in the research literature, the minimum overlap is actually roughly $>90\%$ of the UNL. Finally, we show that during the present stages of diversifying trusted network operators \citep{decentralizeStrategy}, the XRP network is both safe and cannot become ``stuck" making no forward progress.

This research provides a definitive result about the safety of XRP Ledger in its current state. However, to encourage greater flexibility in choosing UNLs in the future, we would prefer an algorithm that gets closer to the original expected overlap bounds. In a sibling paper \citep{Ethan}, we present a novel alternative consensus algorithm called \textbf{Cobalt} that lowers the overlap bound to only $>60\%$ in the general fault model, cannot get stuck in \textit{any} network that satisfies the overlap bound, and has several other properties that make it suitable for eventually replacing XRP LCP. This paper thus serves primarily to show that the XRP Ledger is safe in the interim while transitioning to Cobalt, and the relatively strict requirements on UNL configurations under XRP LCP should be viewed in light of this planned transition.

Section~\ref{definitions} defines the network model and defines the consensus problem. Section~\ref{protocol} is a detailed description of the Ripple consensus algorithm. In section~\ref{analysis}, we prove the network conditions needed to guarantee correctness of the algorithm. Finally, section~\ref{conclusion} concludes with a discussion of the results and directions for improvement.

\section{Network Model and Problem Definition} \label{definitions}

Let $\mathcal{P}_i$ be a node in the network with unique identifier $i$, such as a cryptographic public key. Each node $\mathcal{P}_i$ is free to choose a \textbf{unique node list} $\mathsf{UNL}_i$, which is the set of nodes (possibly including itself) whose messages $\mathcal{P}_i$ will listen to as part of the XRP LCP. The UNL represents a subset of the network which, when taken collectively, is “trusted” by $\mathcal{P}_i$ to not collude in an attempt to defraud the network (see \citep{schwartz2014ripple} for the motivation of this name). The UNLs give structure to the network, with a node that is present in more UNLs implicitly having more influence. An individual node has complete discretion in the choice of their UNL, although we show in section~\ref{analysis} that minimum overlap with other UNLs is necessary for consistency and liveness with other honest nodes. We do not assume that trust is symmetric, so for instance there may be a node $\mathcal{P}_j\in\mathsf{UNL}_i$ such that $\mathcal{P}_i\notin\mathsf{UNL}_j$. Figure \ref{fig:trustgraph} shows an example trust network.

A node that is not crashed and behaves exactly according to the XRP LCP specification is said to be \textbf{honest} or \textbf{correct}; we use the two terms interchangeably. Any node that does not behave according to protocol is said to be \textbf{Byzantine}. Byzantine behavior can include not responding to messages, sending incorrect messages, and even sending different messages to different parties. In section~\ref{analysis}, we initially consider a restriction on the adversary called \textbf{Byzantine accountability}, which states that all nodes -- even Byzantine ones -- cannot send different messages to different nodes. This was part of the original whitepaper \citep{schwartz2014ripple}, which assumed such behavior in a peer-to-peer network could be identified and corrected by honest nodes. However, due to asynchrony and the possibility of honest nodes being temporarily partitioned, this assumption is in practice tenuous at best. Thus the bulk of our results do not depend on this assumption and we will clearly state when it is assumed.

For any node $\mathcal{P}_i$, we denote $n_i=|\mathsf{UNL}_i|$ and define the \textbf{quorum}, denoted $q_i$, a parameter which roughly specifies the minimum number of agreeing nodes in $\mathsf{UNL}_i$ that $\mathcal{P}_i$ needs to hear from to commit to a decision. Each node sets $q_i$ to be $80\%$ of its UNL size, or, more exactly, $q_i=\lceil 0.8 n_i\rceil$.  We assume at most $t_i\leqslant n_i-q_i$ nodes in $\mathsf{UNL}_i$ may be Byzantine faulty.

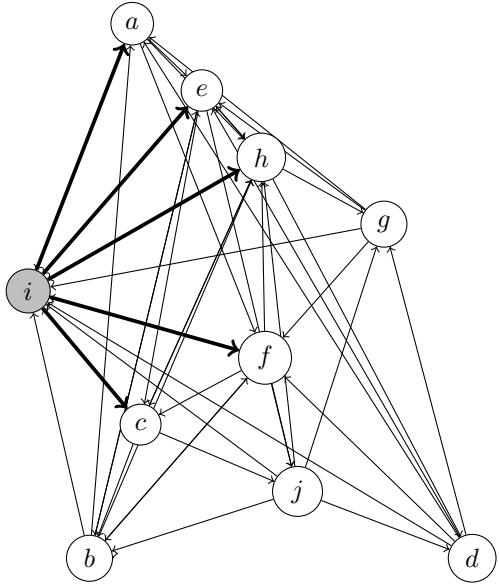
\begin{figure}
    \centering
    \begin{tikzpicture}[scale=0.35]
    \node[peer] (a) at (254.0bp,594.0bp) [] {$a$};
    \node[peer] (c) at (263.0bp,162.0bp) [] {$c$};
    \node[peer] (b) at (208.0bp,18.0bp) [] {$b$};
    \node[peer] (e) at (329.0bp,522.0bp) [] {$e$};
    \node[peer] (d) at (620.0bp,18.0bp) [] {$d$};
    \node[peer] (g) at (525.0bp,378.0bp) [] {$g$};
    \node[peer] (f) at (397.0bp,234.0bp) [] {$f$};
    \node[peer, fill=gray!50] (i) at (142.0bp,306.0bp) [] {$i$};
    \node[peer] (h) at (393.0bp,450.0bp) [] {$h$};
    \node[peer] (j) at (432.0bp,90.0bp) [] {$j$};
    \begin{pgfonlayer}{background}
    \draw [->] (g) -- (e);
    \draw [->] (d) -- (e);
    \draw [->] (h) -- (i);
    \draw [line width=0.5mm,->] (i) -- (a);
    \draw [->] (g) -- (a);
    \draw [->] (j) -- (b);
    \draw [line width=0.5mm,->] (i) -- (h);
    \draw [->] (b) -- (h);
    \draw [->] (c) -- (b);
    \draw [->] (h) -- (e);
    \draw [->] (b) -- (f);
    \draw [->] (e) -- (c);
    \draw [->] (a) -- (i);
    \draw [->] (e) -- (j);
    \draw [->] (f) -- (h);
    \draw [->] (g) -- (f);
    \draw [->] (g) -- (i);
    \draw [->] (j) -- (d);
    \draw [->] (d) -- (f);
    \draw [->] (c) -- (j);
    \draw [->] (b) -- (i);
    \draw [->] (d) -- (g);
    \draw [->] (d) -- (i);
    \draw [->] (e) -- (b);
    \draw [->] (a) -- (h);
    \draw [->] (f) -- (i);
    \draw [->] (f) -- (b);
    \draw [->] (j) -- (h);
    \draw [->] (e) -- (i);
    \draw [->] (b) -- (a);
    \draw [->] (c) -- (i);
    \draw [->] (a) -- (f);
    \draw [line width=0.5mm,->] (i) -- (f);
    \draw [->] (f) -- (j);
    \draw [->] (h) -- (g);
    \draw [->] (f) -- (c);
    \draw [->] (e) -- (h);
    \draw [->] (b) -- (b);
    \draw [->] (c) -- (h);
    \draw [->] (h) -- (c);
    \draw [->] (b) -- (e);
    \draw [->] (j) -- (i);
    \draw [line width=0.5mm,->] (i) -- (c);
    \draw [->] (d) -- (a);
    \draw [->] (a) -- (e);
    \draw [->] (h) -- (d);
    \draw [->] (j) -- (g);
    \draw [line width=0.5mm,->] (i) -- (e);
    \end{pgfonlayer}
    \end{tikzpicture}
    \caption{Example trust graph. The highlighted edges represent the UNL of node $\mathcal{P}_i$.}
    \label{fig:trustgraph}
\end{figure}
Let $L$ represent a ledger, which is the shared state of the system and includes account settings, balances, order books, etc. Two ledgers $L,L'$ are the same if they represent the same ordered history of transactions starting from the unique genesis ledger. Each ledger also has sequence number $\mathsf{seq}(L)$ that is one greater than its parent ledger's sequence number. The genesis ledger has $\mathsf{seq}(L)=1$. A ledger $L$ is created by applying a sequence of transactions $T=[x_0,x_1,\ldots]$ to its parent $\mathsf{parent}(L)$ according to the protocol rules. Two ledgers may have the same parent ledger and sequence number, but differ because they applied different transactions. Note though that the protocol specifies that every set of transactions has a deterministic ordering, so it is not possible for two correct nodes to apply the same transactions but in a different order.

The nodes communicate over a peer-to-peer network which has no prescribed relation with the UNL structure. We simply assume that for every node $\mathcal{P}_i$ and every node $\mathcal{P}_j \in \mathsf{UNL}_i$, there is a reliable authenticated channel for $\mathcal{P}_i$ to receive messages from $\mathcal{P}_j$. To implement such an authenticated channel, all messages are cryptographically signed and verified by receivers. Each node uses a single communication primitive \textbf{broadcast}, which when called from node $\mathcal{P}_j$ sends the same message to all nodes $\mathcal{P}_i$ for which $\mathcal{P}_j \in \mathsf{UNL}_i$. In the algorithms presented in appendix~\ref{appendix}, we use a corresponding \textbf{receive} primitive which is called asynchronously upon the receipt of a broadcast message.

The outcome of XRP LCP is for a node to \textbf{fully validate} ledgers. A fully validated ledger is irrevocable and authoritative, and reflects transactions submitted by network clients and accepted by the consensus algorithm. Fully validating a ledger also fully validates all of its ancestors. In this context, a \textbf{fork} is a situation in which two honest nodes fully validate contradictory ledgers, i.e., different ledgers with the same sequence number. The network is said to be \textbf{fork-safe} if it can never fork with a tolerated configuration of Byzantine nodes.

Although XRP LCP is typically defined in terms of the fully validated chain of ledgers, since each ledger in the chain represents a deterministically-ordered sequence of transactions it also can be considered as an \textbf{atomic broadcast} protocol with batching of transactions for efficiency. Formally, an atomic broadcast protocol is an algorithm in which a set of \textbf{clients}, arbitrarily many of which may be Byzantine faulty, can broadcast \textbf{transactions}, and each node can \textbf{accept} some of those transactions according to the following properties:
\begin{enumerate}
    \item ABC-Agreement: If a correct node accepts a transaction into a ledger, then eventually all correct nodes accept the transaction into a ledger.
    \item ABC-Linearizability: If a correct node accepts transaction $x$ before transaction $x'$, then all correct nodes accept transaction $x$ before $x'$.
    \item ABC-Censorship-Resilience: If a correct client broadcasts a valid transaction $x$ to all correct nodes, $x$ will eventually be accepted by all nodes.
\end{enumerate}

An atomic broadcast algorithm is in particular a variant of consensus \citep{Chandra1996UFD}. Note that Censorship Resilience is the formalized definition of forward progress. In practice, the peer-to-peer network weakens the requirement that clients submit their transaction to all correct nodes, since correct nodes will echo a submitted transaction to each other, flooding the network until every node receives it.

In order to evaluate the correctness of XRP LCP, we model the peer-to-peer network as if it were controlled by a \textbf{network adversary} that can behave arbitrarily. The adversary is controls the delivery order of all messages, as well as at most $t_i$ nodes in $\mathsf{UNL}_i$ for any correct node $\mathcal{P}_i$. We assume though that the adversary is computationally bounded; in particular, it is unable to break generally accepted cryptographic protocols. The identities of Byzantine nodes are unknown in advance by honest nodes in the network.

Since atomic broadcast is a variant of consensus, by the FLP result \citep{Fischer1985IDC} we cannot guarantee forward progress in the presence of arbitrary asynchrony and faulty nodes. Instead, we assume a form of ``weak asynchrony": safety should hold under arbitrary asynchrony, but censorship resilience is only guaranteed to hold under the assumption that the network is eventually \textbf{civil}, meaning that messages are delivered within some protocol-specified maximum delay bound and no nodes are faulty.

In order to help enforce the bounded delay in the XRP ledger implementation, several heuristics are used to identify lagging nodes, prevent excessive flooding of messages and route traffic through the network. Protocol parameters also define maximum delays on different trusted messages in an attempt to aid liveness. Although we ignore these details to simplify the presentation below, we stress that they are important practical considerations in the actual implementation and control the real world performance of the algorithm. The fact that XRP LCP is only weakly asynchronous and its performance depends on these parameters is a limitation of the algorithm. Cobalt, the proposed alternative algorithm to XRP LCP, does not have these limitations. It uses cryptographic randomness to evade the FLP result and guarantees forward progress even with the maximal number of tolerated faulty nodes and unbounded asynchrony \citep{Ethan}.

\begin{table}
    \begin{tabular} { l l}
        $\mathcal{P}_i$ & A node in the network\\
        $\mathsf{UNL}_i$ & The unique node list (UNL) selected by $\mathcal{P}_i$ \\
        $n_i, q_i, t_i$ & The size, validation quorum and maximum \\
        & Byzantine faults of $\mathsf{UNL}_i$\\
        $L$ & A ledger \\
        $\mathsf{seq}(L)$ & Sequence number of ledger $L$\\
        $\hat{L}$ & Fully validated ledger with largest sequence\\
        ~ &  number (the fully validated tip ledger)\\
        $\tilde{L}$ & Current working ledger of deliberation \\
        $T = \{x_0,x_1,\ldots\}$ & A set of transactions\\
        $P_{T,r,L, i}$ & Node $\mathcal{P}_i$'s $r$-th deliberation proposal of\\
        ~ & transactions $T$ to apply to $L$ \\
        $V_{L, i}$ & Node $\mathcal{P}_i$'s validation of ledger $L$\\
        $\mathsf{supp}_{tip}(L),\mathsf{supp}_{branch}(L)$ & Tip and branch support of a ledger $L$\\
        $\mathsf{uncommitted}(s)$ & Uncommitted support at sequence number $s$ \\
        $\phi(L,L')$ & Ordering function that is $1$ if $L > L'$ (by hash)\\
        ~ & and $0$ otherwise
    \end{tabular}
    \caption{Summary of notation}
    \label{table:notation}
\end{table}

Table \ref{table:notation} summarizes our notation, including some that will be explained in subsequent sections. 

\section{The XRP Ledger Consensus Protocol} \label{protocol}
The XRP Ledger Consensus Protocol consists of three primary components:
\begin{itemize}
    \item \textbf{Deliberation}, in which nodes iteratively propose a transaction set to apply to a prior ledger, based on proposals received from other trusted nodes. When a node believes enough proposals agree, it applies the corresponding transactions to the prior ledger according to the ledger protocol rules. It then issues a validation for the generated ledger.
    \item \textbf{Validation}, in which nodes decide whether to fully validate a ledger, based on the validations issued by trusted nodes. Once a quorum of validations for the same ledger is reached, that ledger and its ancestors are deemed fully validated and its state is authoritative and irrevocable. 
    \item \textbf{Preferred Branch}, in which nodes determine the preferred working branch of ledger history. In times of asynchrony, network difficulty, or Byzantine failure, nodes may not initially validate the same ledger for a given sequence number. In order to make forward progress and fully validate later ledgers, nodes use the ledger ancestry of trusted validations to resume deliberating on the network's preferred ledger.
\end{itemize}

In short, for each sequence number $s$, each peer $\mathcal{P}_i$ issues a validation $V_{L,i}$ for the ledger $L$ with $s=\mathsf{seq}(L)$ that it expects will be fully validated by validation. Under civil executions, the deliberation process makes it highly likely the validated ledger will match the ledger validated by its trusted peers. In cases when the network is not working normally, preferred branch ensures peers select a common branch such that nodes will later fully validate the same ledger $L'$ with $\mathsf{seq}(L') > s$. This two-step sequence of deliberation and validation is similar to the proof of stake finality gadget recently introduced by \citet{2017arXiv171009437B}. Indeed, the preferred branch protocol shares a common principle with the GHOST rule of \citet{Sompolinksy2015}.

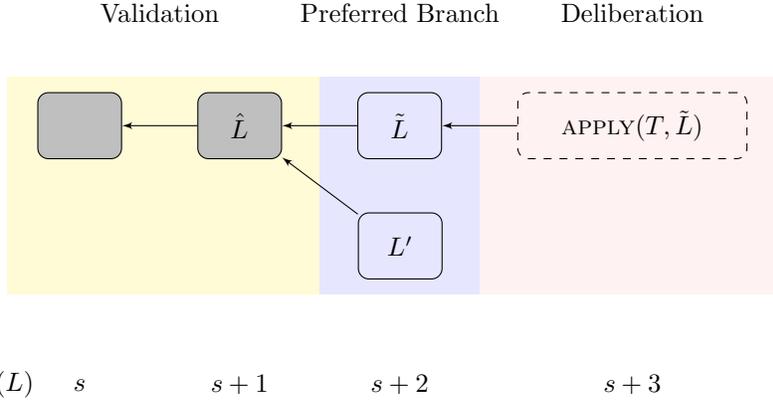
\begin{figure}  
\begin{tikzpicture}[node distance = 1cm, auto]
   \node [ledger, fill=lightgray] (L0) {}; 
    \node [ledger, fill=lightgray,right = of L0] (L1) {$\hat{L}$}; 
    \node [ledger, right = of L1] (L2) {$\tilde{L}$}; 
    \node [ledger, right = of L2,dashed, text width=8em] (L3) {$\textproc{apply}(T,\tilde{L})$}; 
    \node [ledger, below right = of L1, xshift=0.3cm] (L4) {$L'$};
    
    \path [line] (L1) -- (L0);
    \path [line] (L2) -- (L1);
    \path [line] (L3) -- (L2);
    \path [line] (L4) -- (L1);
    \path (L0.east)+(0.5,1.5) node (fv) {Validation};
    \path (L2)+(0,1.5) node (conv) {Preferred Branch};
    \path (L3)+(0,1.5) node (del) {Deliberation};

    \path (L0.south)+(-1.1,-3.0) node (seq) {$\mathsf{seq}(L)$}; 
    \path (L0.south)+(0.0,-3.0) node (s0) {$s$}; 
    \path (L1.south)+(0.0,-3.0) node (s1) {$s+1$}; 
    \path (L2.south)+(0.0,-3.0) node (s2) {$s+2$}; 
    \path (L3.south)+(0.0,-3.0) node (s3) {$s+3$}; 
    \begin{pgfonlayer}{background}
        \path (L0.north west)+(-0.4,0.2) node (a) {};
        \path (L1.east |- L4.south)+(0.5,-0.2) node (b) {};
        \path[fill=yellow!20] (a) rectangle (b);

        \path (L1.north east)+(0.5,0.2) node (a) {};
        \path (L4.south east)+(0.5,-0.2) node (b) {};
        \path[fill=blue!10] (a) rectangle (b);

        \path (L2.north east)+(0.5,0.2) node (a) {};
        \path (L3.south east)+(0.5,-1.8) node (b) {};
        \path[fill=pink!20] (a) rectangle (b);
    \end{pgfonlayer}
\end{tikzpicture}
\caption{Components of XRP LCP. Each rounded rectangle is a ledger with an arrow pointing to its parent.
}
\label{fig:ledgerchain}
\end{figure}

Figure \ref{fig:ledgerchain} is a schematic view of ledger history and shows how these components interact to advance the ledger state from the perspective of a single node. Time flows to the right, with the two grey ledgers on the left defining the fully validated and authoritative ledger chain with tip ledger $\hat{L}$. The dashed ledger on the right represents the deliberation frontier, in which the node is currently negotiating with its trusted nodes on which transactions to apply towards the next ledger. The unfilled ledgers represent two conflicting ledgers $\tilde{L},L'$ that have been validated by different nodes, but which have not received a quorum to fully validate. In this schematic, preferred branch determined the upper ledger $\tilde{L}$ is most likely to be fully validated, so that is the working parent ledger for this node's active deliberation round.

\subsection{Deliberation} \label{protocol:deliberation}
Deliberation is the component of Ripple consensus in which nodes attempt to agree on the set of transactions to apply towards ledgers they validate. Clients submit transactions to one or more nodes in the network, who in turn broadcast the transaction to the rest of the network. Each node maintains a set of these pending transactions that have not been included in a ledger. Starting from this set, a node iteratively proposes new transaction sets based on the support of individual transactions among the sets proposed by nodes in its UNL. Each proposal $P_{T,r,L,i}$ is a tuple of
\begin{itemize}
    \item $T$, the proposing node's current guess of the consensus transaction set.
    \item $r$, the round number of this proposal relative to the other proposals from $\mathcal{P}_i$ based on prior ledger $L$.
    \item $L$, the prior ledger these transactions will apply to.
    \item $i$, the identifier of the node $\mathcal{P}_i$ that broadcasts this proposal.
\end{itemize}
When enough nodes in its UNL propose the same transaction set, a node issues a validation based on that set and begins the next round of deliberation.

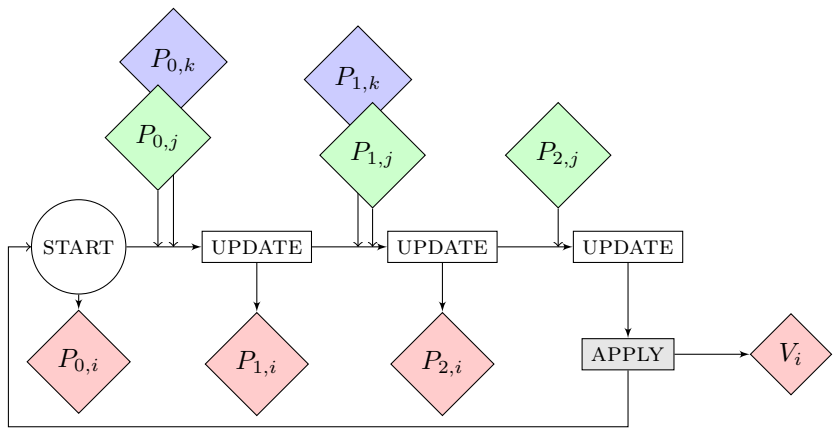
\begin{figure}
    \centering
    \begin{tikzpicture}[node distance = 1cm, auto]
        \node (start) [circle,draw] {\textproc{start}};
        \node (up0) [rectangle, draw, right=of start] {\textproc{update}};
        \node (up1) [rectangle, draw, right=of up0] {\textproc{update}};
        \node (up2) [rectangle, draw, right=of up1] {\textproc{update}};
        \node (p0) [diamond,draw,fill=red!20,below= 0.2 of start] {$P_{0,i}$};
        \node (p1) [diamond,draw,fill=red!20,below= 0.65 of up0] {$P_{1,i}$};
        \node (p2) [diamond,draw,fill=red!20,below= 0.65 of up1] {$P_{2,i}$};        
        \node (nodek0) [diamond,draw,fill=blue!20,above right = 1.65 and 0.45 of start] {$P_{0,k}$};
        \node (nodej0) [diamond,draw,fill=green!20,above right = 0.65 and 0.25 of start] {$P_{0,j}$};
        \node (nodek1) [diamond,draw,fill=blue!20,above right = 1.65 and 0.25 of up0] {$P_{1,k}$};
        \node (nodej1) [diamond,draw,fill=green!20,above right = 0.65 and 0.45 of up0] {$P_{1,j}$};
        \node (nodej2) [diamond,draw,fill=green!20,above right = 0.65 and 0.45 of up1] {$P_{2,j}$};
        \node (apply) [rectangle,draw,fill=gray!20,below=of up2] {\textproc{apply}};
        \node (validate) [diamond,draw,fill=red!20,right=of apply] {$V_{i}$};
        \path[line] (start) -- (up0);
        \path[line] (up0) -- (up1);
        \path[line] (up1) -- (up2);
        \path[line] (start) -- (p0);
        \path[line] (up0) -- (p1);
        \path[line] (up1) -- (p2);
        \path[line] (up2) -- (apply);
        \path[line] (apply) -- (validate);
        \draw [->] ($ (apply.south) $)
                   -- ++(0,-0.75)
                -| ($ (start.west) - (0.3,0) $)
                -- (start.west);
        \begin{pgfonlayer}{background}
            \draw [->] (nodek0.south) -- (start -| nodek0);
            \draw [->] (nodej0.south) -- (start -| nodej0);
            \draw [->] (nodek1.south) -- (start -| nodek1);
            \draw [->] (nodej1.south) -- (start -| nodej1);
            \draw [->] (nodej2.south) -- (start -| nodej2);
        \end{pgfonlayer}
    \end{tikzpicture}
    \caption{Overview of deliberation from the perspective of $\mathcal{P}_i$. The diamonds are proposals or validations, with all but the deliberation round and node identifier subscripts removed for brevity.}
    \label{fig:deliberation}
\end{figure}

Figure~\ref{fig:deliberation} is a high-level overview of deliberation, presented formally in algorithm~\ref{alg:deliberation}. A node initially calls \textproc{start} to begin a new deliberation round and propose its initial view of the consensus transactions. It asynchronously processes new proposals from trusted nodes, maintaining the set of most recent proposals from each. Proposals are only considered if they are for the same prior ledger $\tilde{L}$. The node regularly \textproc{update}s its proposed consensus transaction set in response to newly received node proposals, only including transactions present in at least $threshold(r)$ of the most recently received proposals from its trusted nodes. The threshold starts as a simple majority of the nodes in the UNL, but ratchets up as deliberation rounds proceed. This ensures slow nodes can't prevent consensus converging. In the XRP Ledger implementation, the threshold goes $0.5\to 0.65\to 0.70 \to 0.95$ as $r$ increases. Each node $\mathcal{P}_i$ declares consensus reached when it sees the quorum $q_i$ of its trusted nodes agree on the transaction set. It then applies the consensus transactions to generate the next ledger $L$, broadcasts its validation $V_{L,i}$ and begins a new round of deliberation. 

It is important to note that a node may only validate one ledger with a given sequence number. In fact, the invariant is for node $\mathcal{P}_i$ to only issue a validation $V_{L,i}$ for a ledger $L$ if $\mathsf{seq}(L)$ is greater than that of any ledger previously validated by $\mathcal{P}_i$. Thus if during deliberation, a node determines it is not working on the preferred branch, it will switch to work on the preferred ledger but will not issue a validation until it has caught back up to the sequence number it was on before switching. 

In the XRP Ledger implementation of deliberation, protocol timing parameters determine the synchronization requirements of node proposals and conditions for ending deliberation. Additional waiting periods between phases of deliberation balance the throughput and latency of transaction processing as well as the network overhead of broadcasting proposals and transaction sets. There are also protocol rules that determine which transactions are in the initial proposal and how transactions that failed to be included are retried in subsequent deliberations. Although this deviates from the abstract algorithm presented in this paper, we believe the changes only obscure the presentation of the algorithm and can be viewed as an optimization for increasing transaction throughput. Most importantly, the safety and liveness results in section~\ref{analysis} do not depend on these details of deliberation.

\subsection{Validation} \label{protocol:fullvalidation}
Validation is the simplest of the three components and is summarized in algorithm~\ref{alg:fullvalidation}. Nodes in the network simply listen for validations from trusted nodes. If a node $\mathcal{P}_i$ sees a quorum $q_i$ of validations for a ledger $L$, then it sets the new fully validated tip ledger $\hat{L}$ to $L$. 

\subsection{Preferred Branch} \label{protocol:preferredBranch}
Validators normally validate a simple chain of ledgers, e.g. $L^A\to L^B\to L^C\ldots$. However, during times of asynchrony, network difficulty, or temporary Byzantine failure during deliberation, not all correct nodes may end up receiving enough validations for any individual ledger to fully validate. When presented with conflicting ledgers, preferred branch is the strategy which determines the preferred chain of ledgers to switch to in order to continue making forward progress. It is based on the shared ancestry of the most recent validated ledgers, $lastVals$, and the following quantities:
\begin{enumerate}
    \item The \textbf{tip support} of a ledger $L$, which is the number of trusted nodes whose most recent validated ledger is $L$,
        \begin{equation}
            \mathsf{supp}_{tip}(L) = |\{V_{L',i} \in lastVals : L = L'\}|.
        \end{equation}
    \item The \textbf{branch support} of a ledger $L$, which is the number of trusted nodes whose most recently validated ledger is either $L$ or is descended from $L$,
        \begin{equation}
            \mathsf{supp}_{branch}(L) = \mathsf{supp}_{tip}(L) + |\{V_{L',i} \in lastVals : L \in ancestors(L')\}|,
        \end{equation}
    where $ancestors(L')$ is the set of ancestors of $L'$, i.e. the parent, grandparent, great-grandparent, etc., all the way back to the genesis ledger.
    \item The \textbf{uncommitted support} on a sequence number $s$, which is the number of trusted nodes whose most recent validated ledger is for a ledger with either sequence lower than $s$ or with sequence lower than that of the largest ledger $L$ validation that we personally have broadcasted:
        \begin{equation}
            \mathsf{uncommitted}(s) = |\{V_{L',i} \in lastVals : \mathsf{seq}(L') < \max(s,\mathsf{seq}(L)).
        \end{equation}
\end{enumerate}

Figure~\ref{fig:ledgerAncestry} shows a motivating example, where each ledger is annotated with the tuple of $(\mathsf{supp}_{tip},\mathsf{supp}_{branch}, \mathsf{uncommitted})$ from the perspective of a node that last validated $L^F$. There are 5 trusted nodes, two that last validated $L^F$ and one each validating $L^B$, $L^D$ and $L^E$. The preferred branch strategy determines that $L^D$ is preferred.

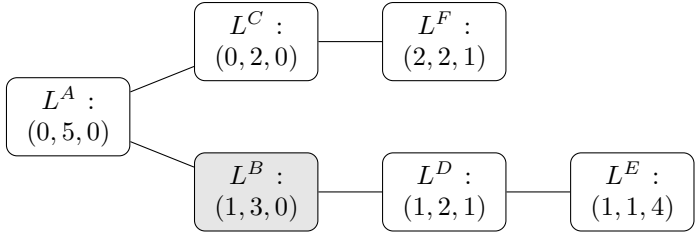
\begin{figure}
    \begin{tikzpicture}[grow=right, level distance=2.5cm, sibling distance=2.0cm]
        \node[wideledger] (LA) {$L^A : (0,5,0) $}
            child { node[wideledger,fill=gray!20] (LB) {$L^B : (1,3,0)$}
                child { node[wideledger] (LD) {$L^D : (1,2,1)$} 
                    child { node[wideledger] (LE) {$L^E : (1,1,4)$}}
                }
            }
            child { node[wideledger] (LC) {$L^C : (0,2,0)$} 
                child { node[wideledger] (LF) {$L^F : (2,2,1)$}
                }
            };
    \end{tikzpicture}
    \caption{Ledger ancestry annotated with tuple of $\mathsf{supp}_{tip}$,$\mathsf{supp}_{branch}$, and $\mathsf{uncommitted})$ from the perspective of the node that last validated $L^F$. Ledger $L^D$ is preferred.}
    \label{fig:ledgerAncestry}
\end{figure}
The preferred branch protocol is provided in Algorithm~\ref{alg:preferredBranch}.  Intuitively, the idea is for each node to be conservative and only switch to a different branch when it knows enough nodes have committed to that chain of ledgers such that an alternative chain cannot have more support. The preferred ledger is found by walking the ancestry tree, starting from the ledger $L$ that is the common ancestor ledger of the most recently validated ledgers. We then select the child ledger $L' \in children(L)$ with highest $\mathsf{supp}_{branch}(L')$, that would still have the most support even if all $\mathsf{uncommitted}(\mathsf{seq}(L'))$ picked a conflicting sibling ledger. If we cannot select a child of $L$ satisfying this requirement, then $L$ is the preferred ledger. If we can find a preferred child $L'$, then we repeat the process on the children of $L'$. In order to ensure a total ordering on ledgers and break ties between sibling ledgers, we rely on a function $\phi(L',L'')$, which is $1$ if the hash of $L' > L''$ and is 0 otherwise. If the ledger found is an ancestor of our current working ledger $\tilde{L}$, we keep $\tilde{L}$ is the preferred ledger, since we do not yet know we are on the wrong branch.

Note that in cases of extreme asynchrony, a branch may only be initially supported by a single node when it becomes preferred. That means other nodes must verify the protocol invariants of that ledger before switching to deliberate on it.

\section{Analysis} \label{analysis}

Having described the XRP LCP from a protocol perspective, we now formally prove results about its safety and liveness.

\subsection{Safety}

In this section we will prove conditions on the network configuration that guarantee different nodes running the XRP LCP will remain consistent.

In the following analysis, we find it convenient to assume that fully validating a ledger \textit{does not} fully validate its ancestors. As we will find, it turns that it may be possible in some configurations for two nodes to fully validate contradictory ledgers with different sequence numbers even if fully validating contradictory ledgers with the same sequence number is impossible. The former is just as problematic as the latter, but we find it convenient to separate out the two failure types and prove conditions preventing them separately.

For the initial analysis, we make a simplifying assumption. Later in the paper we will reanalyze the problem without these assumptions. The assumption we make is that a Byzantine faulty node cannot convince two honest nodes that it validated different ledgers. This assumption was used in the original whitepaper \citep{schwartz2014ripple} and is rationalized by the idea that all communication is done through generic multicast over a peer-to-peer network, so that the ``echoes" of two contradictory messages would be noticed in time for the message to be ignored, since the messages are signed. Unfortunately this assumption does not hold in a fully asynchronous network, since a network partition could segregate the contradictory messages long enough for damage to occur. We call this assumption \textbf{Byzantine accountability}, and will always reference when we are using it as an assumption.

Assuming Byzantine accountability holds, we now analyze when it is possible for two nodes to fully validate different ledgers in a single round of consensus. Modifying the notation slightly, the condition suggested in the whitepaper \citep{schwartz2014ripple} was that two nodes $\mathcal{P}_i,\mathcal{P}_j$ cannot fully validate conflicting ledgers if
	\begin{align*}
	|\mathsf{UNL}_i\cap\mathsf{UNL}_j|\geqslant \max\{n_i-q_i,n_j-q_j\}.
	\end{align*}
With quorums of $80\%$ as suggested, this condition is more easily identified with the actual condition given in the whitepaper,
	\begin{align*}
	|\mathsf{UNL}_i\cap\mathsf{UNL}_j|\geqslant 0.2\max\{n_i,n_j\}.
	\end{align*}

In a later independent analysis, Armknecht et al. \citep{Armknecht2015} showed that this condition is incorrect. They instead suggest that $\mathcal{P}_i,\mathcal{P}_j$ cannot fully validate different ledgers if \textit{and only if}
	\begin{align*}
	|\mathsf{UNL}_i\cap\mathsf{UNL}_j|>2\max\{n_i-q_i,n_j-q_j\}.
	\end{align*}
It is true that if the above condition holds then $\mathcal{P}_i$ and $\mathcal{P}_j$ cannot fully validate different ledgers. However, the converse is not true as the following proposition shows.

\begin{proposition}\label{basicOverlap}
Assuming Byzantine accountability, two honest nodes $\mathcal{P}_i,\mathcal{P}_j$ cannot fully validate different ledgers with the same sequence number iff
	\begin{align*}
	|\mathsf{UNL}_i\cap\mathsf{UNL}_j|>n_i-q_i+n_j-q_j.
	\end{align*}
\end{proposition}

Note that $2\max\{n_i-q_i,n_j-q_j\}\geqslant n_i-q_i+n_j-q_j$, but $n_i-q_i+n_j-q_j$ is strictly smaller whenever $n_i-q_i\neq n_j-q_j$. Thus the condition suggested by Armknecht et al. is sufficient but not necessary.

\begin{proof}[Proof of proposition \ref{basicOverlap}]
We first prove sufficiency. Suppose $\mathcal{P}_i$ fully validates the ledger $L$ and $|\mathsf{UNL}_i\cap\mathsf{UNL}_j|>n_i-q_i+n_j-q_j$.

Let $S$ be the set of nodes in $\mathsf{UNL}_i$ that validated $L$. Since $\mathcal{P}_i$ fully validated $L$, $|S|\geqslant q_i$. By Byzantine accountability, every node in $S\cap\mathsf{UNL}_j$ could not have sent a validation to $\mathcal{P}_j$ for any ledger $L'\neq L$. Thus it suffices to show that $|S\cap\mathsf{UNL}_j|>n_j-q_j$, since then there cannot be any ledger $L'\neq L$ with $q_j$ support in $\mathsf{UNL}_j$.

By the overlap hypothesis, we have
	\begin{align*}
	|S\cap\mathsf{UNL}_j|&=|S|-|S\setminus\mathsf{UNL}_j|\\
	&\geqslant |S|-|\mathsf{UNL}_i\setminus\mathsf{UNL}_j|\\
	&=|S|-\left(|\mathsf{UNL}_i|-|\mathsf{UNL}_i\cap\mathsf{UNL}_j|\right)\\
	&>|S|-\left(n_i-\left(n_i-q_i+n_j-q_j\right)\right)\\
	&\geqslant q_i-\left(n_i-\left(n_i-q_i+n_j-q_j\right)\right)\\
	&=n_j-q_j.
	\end{align*}

For necessity, first suppose $|\mathsf{UNL}_i\cap\mathsf{UNL}_j|\leqslant n_j-q_j$. Then all the nodes in $\mathsf{UNL}_i$ can validate $L$, while all the nodes in $\mathsf{UNL}_j\setminus\mathsf{UNL}_i$ validate $L'$, and by assumption $|\mathsf{UNL}_j\setminus\mathsf{UNL}_i|\geqslant q_j$ so $\mathcal{P}_j$ fully validates $L'$ while $\mathcal{P}_i$ fully validates $L$.

Now suppose $|\mathsf{UNL}_i\cap\mathsf{UNL}_j|\leqslant n_i-q_i+n_j-q_j$ and $|\mathsf{UNL}_i\cap\mathsf{UNL}_j|>n_j-q_j$. Then
	\begin{align*}
	|\mathsf{UNL}_i\setminus\mathsf{UNL_j}|&=|\mathsf{UNL}_i|-|\mathsf{UNL}_i\cap\mathsf{UNL_j}|\\
	&\geqslant n_i-\left(n_i-q_i+n_j-q_j\right)\\
	&=q_i+q_j-n_j.
	\end{align*}

Thus if all nodes in $\mathsf{UNL}_i\setminus\mathsf{UNL_j}$ validate $L$ and $n_j-q_j$ nodes in $\mathsf{UNL}_i\cap\mathsf{UNL}_j$ also validate $L$ (which is possible since $|\mathsf{UNL}_i\cap\mathsf{UNL}_j|>n_j-q_j$ by assumption), then $\mathcal{P}_i$ will receive $(q_i+q_j-n_j)+(n_j-q_j)=q_i$ validations for $L$ and fully validate $L$. Meanwhile, if all the other nodes in $\mathsf{UNL}_j$ validate $L'$, then since only $n_j-q_j$ nodes in $\mathsf{UNL}_j$ validated $L$, $n_j-(n_j-q_j)=q_j$ nodes will validate $L'$, so $\mathcal{P}_j$ will fully validate $L'$.
\end{proof}

Assuming a quorum of $80\%$, these overlap conditions may be summarized as follows:
\begin{itemize}
	\item Schwartz et al.: Every pair of nodes needs an overlap of $20\%$ the \textit{maximum} size of their respective UNLs.
	\item Armknecht et al.: Every pair of nodes needs an overlap of $41\%$ the \textit{maximum} size of their respective UNLs.
	\item Actual condition: Every pair of nodes needs an overlap of $41\%$ of the \textit{average} size of their respective UNLs.
\end{itemize}

For the remainder of the paper we will no longer assume Byzantine accountability. In a live network, one would prefer absolute safety rather than relying on brittle heuristics that suggest it is unlikely that a Byzantine node could send conflicting messages to different nodes without getting caught. Thus we follow the research convention and assume that Byzantine nodes can send arbitrary messages to arbitrary nodes.

For any pair of nodes $\mathcal{P}_i$ and $\mathcal{P}_j$, let $\mathsf{O}_{i,j}=|\mathsf{UNL}_i\cap\mathsf{UNL}_j|$ and let $t_{i,j}=\min\{t_i,t_j,\mathsf{O}_{i,j}\}$. $t_{i,j}$ is the maximum number of allowed Byzantine faults in $\mathsf{UNL}_i\cap\mathsf{UNL}_j$, assuming that there are at most $t_i$ faults in $\mathsf{UNL}_i$ and at most $t_j$ faults in $\mathsf{UNL}_j$.

The following lemma will be useful throughout the paper.

\begin{lemma}\label{mValidations}
If an honest node $\mathcal{P}_i$ sees $m$ validations for the ledger $L$ with $\mathsf{seq}(L)=s$, then for any other honest node $\mathcal{P}_j$, there are at least $\mathsf{O}_{i,j}+m-n_i-t_{i,j}$ honest nodes in $\mathsf{UNL}_j$ that validated $L$. Furthermore, there can be \textit{exactly} $\mathsf{O}_{i,j}+m-n_i-t_{i,j}$ honest nodes in $\mathsf{UNL}_j$ that validated $L$.
\end{lemma}
\begin{corollary}\label{mValidationsInverse}
If an honest node $\mathcal{P}_i$ sees $m$ validations for the ledger $L$ with $\mathsf{seq}(L)=s$, then $\mathcal{P}_j$ can see at most $n_i+n_j-\mathsf{O}_{i,j}-m+t_{i,j}$ validations for any contradictory ledger $L'$ with $\mathsf{seq}(L')=s$. Furthermore, it is possible for $\mathcal{P}_j$ to see exactly $n_i+n_j-\mathsf{O}_{i,j}-m+t_{i,j}$ validations for a contradictory ledger $L'$ with sequence number $s$.
\end{corollary}
\begin{proof}
If an honest node validates $L$, then $\mathcal{P}_j$ cannot receive a validation for any contradictory ledger $L'$ from it. By lemma \ref{mValidations}, there are at least $\mathsf{O}_{i,j}+m-n_i-t_{i,j}$ honest nodes that validate $L$. If every other node in $\mathsf{UNL}_j$ sends a validation to $\mathcal{P}_j$ for some contradictory ledger $L'$, then $\mathcal{P}_j$ can receive up to (and including, by the secondary clause of lemma \ref{mValidations})
\begin{align*}
n_j-\left(\mathsf{O}_{i,j}+m-n_i-t_{i,j}\right)=n_i+n_j-\mathsf{O}_{i,j}-m+t_{i,j}
\end{align*}
validations for $L'$ with $\mathsf{seq}(L')=s$.
\end{proof}

Note that corollary \ref{mValidationsInverse} does not preclude the possibility that $\mathcal{P}_j$ will see more than $n_i+n_j-\mathsf{O}_{i,j}-m+t_{i,j}$ validations for a contradictory ledger with a larger sequence number than $s$. Indeed, without assuming totality, it turns out that such an occurrence is possible, which forces the algorithm to use much tighter safety margins.

\begin{proof}[Proof of lemma \ref{mValidations}]
The proof is similar to the proof of proposition \ref{basicOverlap}.

Suppose $\mathcal{P}_i$ sees $m$ validations for $L$. Again letting $S$ be the set of nodes in $\mathsf{UNL}_i$ that sent validations to $\mathcal{P}_i$ for $L$, then
	\begin{align*}
	|S\cap\mathsf{UNL}_j|&=|S|-|S\setminus\mathsf{UNL}_j|\\
	&\geqslant |S|-|\mathsf{UNL}_i\setminus\mathsf{UNL}_j|\\
	&=|S|-(n_i-\mathsf{O}_{i,j})\\
	&=m-n_i+\mathsf{O}_{i,j}.
	\end{align*}

There could be $t_{i,j}$ Byzantine nodes in $S\cap\mathsf{UNL}_j$ that send $\mathcal{P}_j$ a validation for something other than $L$, so at least $m-n_i+\mathsf{O}_{i,j}-t_{i,j}$ \textit{honest} nodes validated $L$.

For the second point, assume that every node in $\mathsf{UNL}_i\setminus\mathsf{UNL}_j$ validates $L$, $t_{i,j}$ Byzantine nodes in $\mathsf{UNL}_i\cap\mathsf{UNL}_j$ send a validation for $L$ to $\mathcal{P}_i$ and $L'$ to $\mathcal{P}_j$. Then there are exactly $m-n_i+\mathsf{O}_{i,j}-t_{i,j}$ honest nodes in $\mathsf{UNL}_i\cap\mathsf{UNL}_j$ that send a validation for $L$ to $\mathcal{P}_j$. Since every node in $\mathsf{UNL}_j\setminus\mathsf{UNL}_i$ can validate some ledger other than $L$, there can be exactly $m-n_i+\mathsf{O}_{i,j}-t_{i,j}$ honest nodes in $\mathsf{UNL}_j$ that send a validation for $L$ to $\mathcal{P}_j$.
\end{proof}

\begin{proposition}\label{bftImmediateOverlap}
$\mathcal{P}_i$ fully validating some ledger $L$ with $\mathsf{seq}(L)=s$ implies that $\mathcal{P}_j$ cannot fully validate any contradictory ledger with the same sequence number $s$ iff $\mathsf{O}_{i,j}>(n_i-q_i)+(n_j-q_j)+t_{i,j}$.
\end{proposition}
\begin{proof}
By letting $m=q_i$, corollary \ref{mValidationsInverse} tells us that $\mathcal{P}_i$ fully validating $L$ implies that $\mathcal{P}_j$ can see at most $n_i+n_j-\mathsf{O}_{i,j}-q_i+t_{i,j}$ validations for any contradictory ledger with sequence number $s$.

Thus if
\begin{align*}
q_j&>n_i+n_j-\mathsf{O}_{i,j}-q_i+t_{i,j}\\
\mathsf{O}_{i,j}&>n_i-q_i+n_j-q_j+t_{i,j}
\end{align*}
then $\mathcal{P}_j$ cannot fully validate any contradictory ledger with seqeuence number $s$.

For necessity, if $\mathsf{O}_{i,j}\leqslant(n_i-q_i)+(n_j-q_j)+t_{i,j}$, then the second clause of corollary \ref{mValidationsInverse} implies that $\mathcal{P}_j$ can see exactly
\begin{align*}
n_i+n_j-\mathsf{O}_{i,j}-q_i+t_{i,j}&\geqslant n_i+n_j-\left((n_i-q_i)+(n_j-q_j)+t_{i,j}\right)-q_i+t_{i,j}\\
&=q_j
\end{align*}
validations for a contradictory ledger $L'$, allowing $\mathcal{P}_j$ to fully validate $L'$.
\end{proof}

Once again assuming $80\%$ quorums and $20\%$ fault tolerance as in the whitepaper, this overlap condition can be summarized as requiring roughly $61\%$ UNL overlaps.

To see why the overlap hypothesis in proposition \ref{bftImmediateOverlap} does not guarantee full safety, note that it is possible for a node to exit from deliberation for sequence $s$ and then be unable to fully validate any ledger with sequence $s$, as the following example shows.

\begin{figure}  
    \begin{tikzpicture}[align=left, node distance=0.2cm and 0.3cm,
        apeer/.style={circle, draw, fill=white, minimum size=0.8cm},
        transX/.style={diamond,draw,fill=red!20},
        transY/.style={diamond,draw,fill=blue!20,scale=0.9},
        box/.style={rectangle,draw,fill=white}]
        \foreach \x in {1,...,10}
            \node [apeer] (P\x) at (\x, 0)  {\x};


         \foreach \x in {1,...,5}
            \node [transX] (T1\x) at (\x, -1.1)  {$T$};
         \foreach \x in {6,...,10}
            \node [transY] (T1\x) at (\x, -1.1)  {$T'$};

         \node [box] (VX1) at (3,-2.2) {$x_0 : 100\%$\\$x_1 : 50\%$};
         \node [box] (VY1) at (8,-2.2) {$x_0 : 100\%$\\$x_1 : 40\%$};
         \node (Rlabel) at (0.5,-1.8) {Round,$threshold$};
         \node (R1label) at (0,-2.4) {$r-1,50\%$};
         \path[-] (0,-2.85) edge (11,-2.85);
         \begin{pgfonlayer}{background}
            \foreach \x in {1,...,10}
                \draw[->,gray] (T1\x) -- (VX1);
            \foreach \x in {1,2,3,4,6,7,8,9,10}
                \draw[->,gray] (T1\x) -- (VY1);
         \end{pgfonlayer}{background}
         \foreach \x in {1,...,5}
            \node [transX] (T2\x) at (\x, -3.4)  {$T$};
         \foreach \x in {6,...,10}
            \node [transY] (T2\x) at (\x, -3.4)  {$T'$};
        
         \node [box] (VX2) at (3,-4.6) {$x_0 : 100\%$\\$x_1 : 50\%$};
         \node [box] (VY2) at (8,-4.6) {$x_0 : 100\%$\\$x_1 : 50\%$};
         \node (R2label) at (0,-4.8) {$r,50\%$};

         \path[-] (0,-5.25) edge (11,-5.25);
         \begin{pgfonlayer}{background}
            \foreach \x in {1,...,10}
             {
                \draw[->,gray] (T2\x) -- (VX2);
                \draw[->,gray] (T2\x) -- (VY2);
             }
         \end{pgfonlayer}

         \foreach \x in {1,...,10}
            \node [transX] (T3\x) at (\x, -5.8)  {$T$};
        \node [box] (VX3) at (3,-7.0) {$x_0 : 100\%$\\$x_1 : 50\%$};
        \node [box] (VY3) at (8,-7.0) {$x_0 : 100\%$\\$x_1 : 60\%$};
        \node (R3label) at (0,-7.2) {$r+1,65\%$};
        \path[-] (0,-7.65) edge (11,-7.65);
        \begin{pgfonlayer}{background}
            \foreach \x in {2,...,10}
                \draw[->,gray] (T2\x) -- (VX3);
            \draw[->,gray] (T31) -- (VX3);
            \foreach \x in {1,...,9}
                \draw[->,gray] (T2\x) -- (VY3);
            \draw[->,gray] (T310) -- (VY3);
         \end{pgfonlayer}{background}
         \foreach \x in {1,...,10}
            \node [transY] (T4\x) at (\x, -8.2)  {$T'$};
        \node (R3label) at (0,-9.0) {$r+2,65\%$};

         \begin{pgfonlayer}{background}
            \node[draw=red, dashed, very thick,fit=(P1) (P5), inner sep=2.2](UNLX) {};
            \node[text=red,above left] at (UNLX.north west) {$X$};
            \node[draw=blue, dashed, very thick,fit=(P6) (P10), inner sep=2.2](UNLY) {};
            \node[text=blue,above right] at (UNLY.north east) {$Y$};
            \node[draw=magenta, dashed,very thick,fit=(T31) (T36),inner sep=-0.1](UNLZ) {};
            \node[text=magenta,left] at (UNLZ.west) {$Z$};
            \node[draw=brown, dashed,very thick,fit=(T47) (T410),inner sep=-0.1](UNLZp) {};
            \node[text=brown,right] at (UNLZp.east) {$Z'$};
         \end{pgfonlayer}
    \end{tikzpicture}
    \caption{Schematic of example~\ref{deliberationEarlyExit}. The two node groups $X$ and $Y$ begin by proposing $T=\{x_0,x_1\}$ and $T'=\{x_0\}$ respectively. The left (right) boxes reflect proposals seen by node 1 (10) and are representive of all nodes in group $X$ ($Y$). Gray arrows indicate proposals that were received in time to calculate the thresholds for a given round. Note that in round $r-1$, the proposal from node 5 was not received by nodes in $Y$. The two partitions $Z$ and $Z'$ are represented by the dashed boxes. Note that all nodes share the same single UNL.}
    \label{fig:deliberationEarlyExit}
\end{figure}
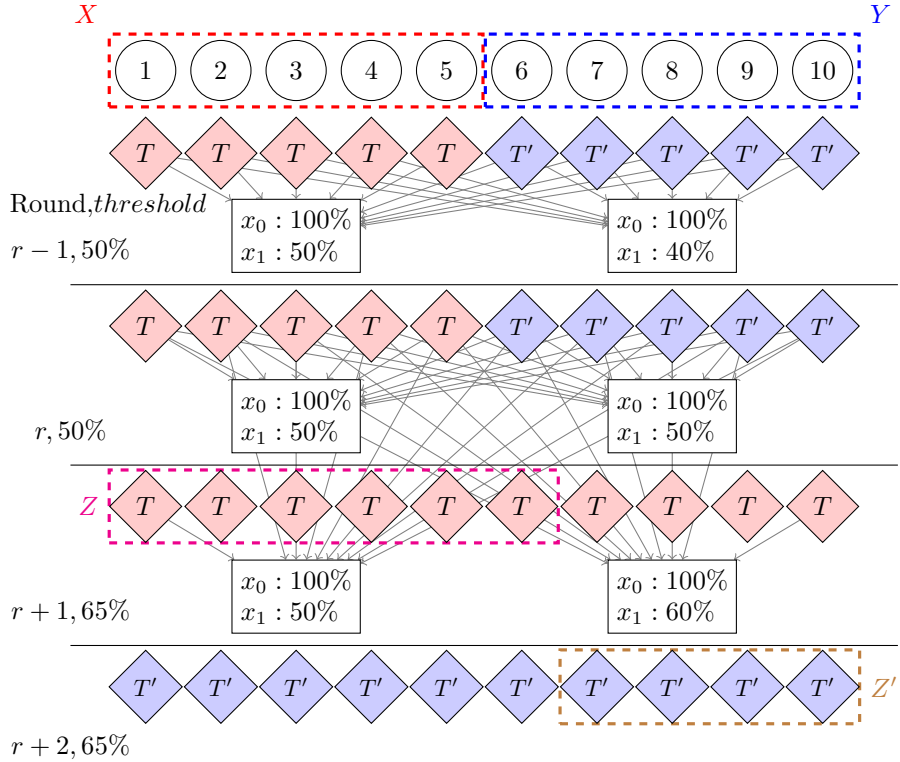

\begin{example}\label{deliberationEarlyExit}
\end{example} Consider a complete network with $10$ nodes. Let $X$ denote the first $5$ nodes and $Y$ denote the other $5$ nodes. As shown in figure \ref{fig:deliberationEarlyExit}, suppose all the nodes in $X$ begin deliberation proposing the transaction set $T=\{x_0,x_1\}$ and all the nodes in $Y$ begin deliberation proposing the set $T'=\{x_0\}$. Thus for a transaction threshold of $\tau$, receiving $\tau$ proposals for $T$ will cause an honest node to propose $T$, while receiving less that $\tau$ proposals for $T$ but $\tau$ proposals for either $T$ \textit{or} $T'$ will cause an honest node to propose $T'$, since $T \cap T' = T'$.

Let $r$ be such that the transaction threshold is $50\%$ in deliberation round $r$ and $65\%$ in round $r+1$. By the ratcheting threshold protocol described in subsection~\ref{protocol:deliberation}, such an $r$ exists. During the first $r-1$ rounds, the nodes in $X$ receive all the proposals, while the nodes in $Y$ receive all the proposals except for one proposal from a node in $X$. Since the nodes in $X$ propose $T$ while the nodes in $Y$ propose $T'$, the nodes in $X$ continue proposing $T$, while the nodes in $Y$ only receive $4$ proposals for $T$ and continue proposing $T'$ by the assumptions on $T$ and $T'$.

Now in round $r$ all nodes receive all proposals. This causes all nodes to propose $T$ in round $r+1$. But in round $r+1$, a network failure causes none of the nodes to receive anyone else's proposals. Based on the most recently received deliberation proposals, everyone assumes that all the other nodes are still proposing what they proposed in the previous round. Thus the nodes in $X$ see only $5$ proposals for $T$ while the nodes in $Y$ see only $6$ proposals for $T$ (since they of course receive their own updated proposals). No one sees $65\%$ support for $x_1$, but everyone sees $100\%$ support for $x_0$, so everyone proposes $T'$ in round $r+2$.

Now pick an arbitrary partition of the network into two sets, $Z$ and $Z'$. The nodes in $Z$ receive all the proposals from round $r+1$ but none of the proposals from round $r+2$. Thus they see $100\%$ support for $T$ and validate $T$. Meanwhile the nodes in $Z'$ receive all the proposals from round $k+2$, see $100\%$ support for $T'$, and validate $T'$. Thus we can exit from deliberation with two arbitrary subsets validating different ledgers. In this case if $|Z|>2$ and $|Z'|>2$, then none of the nodes will fully validate a ledger during this consensus round.

\vspace*{1em}

Because of examples like example \ref{deliberationEarlyExit}, we therefore make the broad assumption that deliberation can terminate with an arbitrary result. In practice, this may require a significantly degraded network, but is nonetheless a real risk. From a theoretical perspective, deliberation is therefore completely irrelevant; it is purely an optimization that makes it so that during civil executions most nodes will go into validation with the same ledger, allowing every node to fully validate usually, and it could be removed without fundamentally changing the protocol.

Thus we shift our focus towards validation without making any assumptions about the result of deliberation. We need to prove that if any node fully validates a ledger $L$, then it is never possible for a node to fully validate a ledger $L'$ such that $\mathsf{seq}(L')\geqslant\mathsf{seq}(L)$ and $L'$ is not a descendant of $L$. The following lemma provides the route for guaranteeing this.

\begin{lemma}\label{preferredBranchStability}
If for every node $\mathcal{P}_i$, there are more than $n_i/2$ honest nodes in $\mathsf{UNL}_i$ that submit a validation for some ledger $L$ with $\mathsf{seq}(L)$, then for every node $\mathcal{P}_i$ more than $n_i/2$ honest nodes in $\mathsf{UNL}_i$ will \emph{always} submit validations for ledgers descended from $L$.
\end{lemma}
\begin{proof}
We prove this by contradiction. If the lemma is not true, then under the lemma's hypotheses there must be some honest node which submits a validation for $L$ and also eventually submits a validation for some ledger $L'$ with $\mathsf{seq}(L')>\mathsf{seq}(L)$ and $L'$ not descended from $L$. Since correct nodes can only submit new validations for ledgers with sequence strictly greater than any ledger they have previously submitted a validation for, such a node must submit its validation for $L'$ \textit{after} having submitted its validation for $L$. Thus let $\mathcal{P}_i$ be the honest node that submitted a validation for $L$ and later is the first to submit a validation for a ledger $L'$ not descended from $L$.

The only way for $\mathcal{P}_i$ to later submit a validation for a ledger off of the $L$ branch is if it runs the preferred branch algorithm and sees a ledger off the $L$ branch as preferred. For a given $s\leqslant \mathsf{seq}(L)$, let $\mathsf{parent}(s,L)$ denote the ancestor of $L$ with sequence $s$. Since $\mathcal{P}_i$ submitted a validation for $L$ by assumption, it considers all validations for ledgers with sequence below $\mathsf{seq}(L)$ as uncommitted in the preferred branch protocol. But since $\mathcal{P}_i$ is assumed to be the first to switch away from the $L$ branch, more than $n_i/2$ nodes in $\mathsf{UNL}_i$ cannot have sent out a validation for any ledger $L'$ with $\mathsf{seq}(L')\geqslant s$ and $L'$ not descended from $L$. Thus for every $s\leqslant \mathsf{seq}(L)$, $\mathcal{P}_i$ sees a majority of nodes in $\mathsf{UNL}_i$ as being either uncommitted support at $s$ or branch support for $\mathsf{parent}(s,L)$. In other words, for all $s\leqslant \mathsf{seq}(L)$,
\begin{align*}
\mathsf{supp}_{branch}(\mathsf{parent}(s,L))>n_i/2-\mathsf{uncommitted}(s).
\end{align*}

For a given $s<\mathsf{seq}(L)$, suppose $\mathsf{parent}(s,L)$ is the current base ledger in the loop on line $10$ of Algorithm~\ref{alg:preferredBranch}. Then either $C[0]=\mathsf{parent}(s+1,L)$ or there is some $L'\neq \mathsf{parent}(s+1,L)$ with $\mathsf{parent}(L')=\mathsf{parent}(s,L)$ and $C[0]=L'$. In the latter case, the branch support for $C[1]$ must be at least equal to the branch support for $\mathsf{parent}(s+1,L)$ (breaking ties with $\phi$) by definition of the ordering of $C$. Further, the branch support for $C[0]$ must be less than $n_i/2$. Thus in line $16$,
\begin{align*}
	\Delta &= \mathsf{supp}_{branch}(C[0])-\mathsf{supp}_{branch}(C[1])+\phi(C[0],C[1]) \\
	&\leqslant \mathsf{supp}_{branch}(C[0])-\mathsf{supp}_{branch}(\mathsf{parent}(s+1,L))+1 \\
	&< n_i/2-\mathsf{supp}_{branch}(\mathsf{parent}(s+1,L))+1 \\
	&< n_i/2-(n_i/2-\mathsf{uncommitted}(s+1))+1 \\
	&\leqslant \mathsf{uncommitted}(s+1)+1,
\end{align*}
so the condition $\Delta>\mathsf{uncommitted}(s+1)$ is always false. Thus in the latter case $\mathcal{P}_i$ sees $\mathsf{parent}(s,L)$ as the preferred ledger. In the former case, $\mathcal{P}_i$ either sees $\mathsf{parent}(s+1,L)$ as the preferred ledger or continues the loop with $\mathsf{parent}(s+1,L)$ as the base ledger. By induction, $\mathcal{P}_i$ is guaranteed to see some ledger on the $L$ branch as preferred, so $\mathcal{P}_i$ cannot leave the $L$ branch, contradicting our assumption about $\mathcal{P}_i$.
\end{proof}

If more than $n_i/2$ honest nodes in $\mathsf{UNL}_i$ only ever validate descendants of $L$, then certainly $\mathcal{P}_i$ cannot fully validate a ledger that doesn't descend from $L$, since otherwise there would be $q_i>n_i/2$ nodes that sent validations for some ledger $L'$ with sequence number $s'$ that \textit{doesn't} descend from $L$. Thus we can show that consensus is safe if we can guarantee that if any honest node fully validates a ledger $L$ with sequence number $s$, then for every node $\mathcal{P}_i$, more than $n_i/2$ honest nodes in $\mathsf{UNL}_i$ must have validated $L$. The following proposition gives the overlap condition guaranteeing this property.

\begin{proposition}\label{bftFullOverlap}
Given two honest nodes $\mathcal{P}_i, \mathcal{P}_j$, $\mathcal{P}_i$ fully validating a ledger $L$ with $\mathsf{seq}(L)=s$ implies that there are more than $n_j/2$ honest nodes in $\mathsf{UNL}_j$ which validated $L$ iff $\mathsf{O}_{i,j}>n_j/2+n_i-q_i+t_{i,j}$.
\end{proposition}
\begin{proof}
The proof is directly analogous to the proof of proposition \ref{bftImmediateOverlap}, except rather than bounding the formula by $q_j$ we bound it by $n_j/2$.
\end{proof}

\begin{theorem}\label{safetyTheorem}
XRP LCP guarantees fork safety if $\mathsf{O}_{i,j}>n_j/2+n_i-q_i+t_{i,j}$ for every pair of nodes $\mathcal{P}_i,\mathcal{P}_j$.\hfill\qed
\end{theorem}

Note that although proposition \ref{bftFullOverlap} is an iff statement, the overlap condition in theorem \ref{safetyTheorem} is only sufficient but not necessary for XRP LCP safety. This is because lemma \ref{preferredBranchStability} is not an iff statement. Further, there may be some validation configurations that cannot come out of deliberation, breaking our broad assumption that anything can come out of deliberation. However, it is the weakest condition that can be expressed purely as a bound on the size of overlaps.

Once again assuming $80\%$ quorums and $20\%$ faults, the overlap condition in theorem \ref{safetyTheorem} can be summarized as requiring roughly $>90\%$ UNL overlaps. Although quite a narrow margin (and certainly far more narrow than originally expected), this does still allow a small amount of variation, which is very important for the XRP Ledger network's transition to a recommended UNL comprised of independent entities. Having some flexibility in the UNLs is important both for after the diversification of trusted operators (as one can never guarantee total agreement on participants when the participants are independent entities) and also during the diversification process (if tiny disagreements during changes to the UNL list could cause a fork, then diversification would always be too risky to execute).

\subsection{Liveness}

Now that we have a concrete metric of when it is impossible for the network to fork, we would like to know when it makes forward progress. If a live network stops making forward progress, that is almost as damaging as forking, since businesses might be relying on being able to make transfers on time. Unfortunately, by the FLP result \citep{Fischer1985IDC} it is impossible to guarantee forward progress in a fully asynchronous network.

In the absence of being able to prove that the network always makes forward progress, we would like to at least be able to prove that the network cannot get ``stuck''. In other words, that the network cannot get into a state in which some honest nodes can never fully validate a new ledger.

Unfortunately, it is very difficult in general to guarantee forward progress with XRP LCP. The following example shows that it is possible to get stuck even with $99\%$ UNL overlaps and no Byzantine faults.

\begin{example}\label{stuckExample}
\end{example}Consider a network of 102 peers drawin in figure~\ref{fig:stuckExample}. There are two UNLs, the red $X=\{\mathcal{P}_1,\mathcal{P}_2,\ldots,\mathcal{P}_{101}\}$ and blue $Y=\{\mathcal{P}_2,\mathcal{P}_3,\ldots,\mathcal{P}_{102}\}$. Peers $1-51$ use $X$ and peers $52-102$ use $Y$. There are two ledgers, $L$ and $L'$. The nodes listening to $X$ all validate a descendant of $L$, while the nodes listening to $Y$ all validate a descendant of $L'$. Since $51>0.5|X|$ nodes in $X$ validate a descendant of $L$. Thus according to the preferred branch protocol all, the nodes listening to $X$ cannot switch branch to $L'$. Similarly, since $51>0.5|Y|$ nodes in $Y$ all validate a descendant of $L$, the nodes listening to $Y$ cannot switch branch to $L'$. The network cannot ever rejoin without manual intervention.

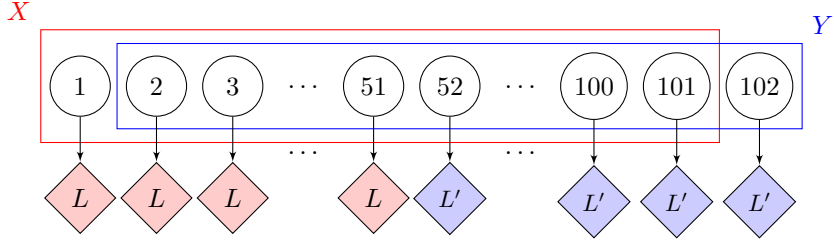
\begin{figure}
    \begin{tikzpicture}[align=center, node distance=0.6cm and 0.2cm,
        apeer/.style={circle, draw, fill=white, minimum size=0.8cm},
        valL/.style={diamond,draw,fill=red!20},
        valLp/.style={diamond,draw,fill=blue!20, scale=0.9}]
         \node [apeer] (P1) {1}; 
         \node [apeer, right =of P1] (P2) {2}; 
         \node [apeer, right =of P2] (P3) {3}; 
         \node [right =of P3] (Pdots) {$\ldots$}; 
         \node [apeer, right =of Pdots] (P51) {51}; 
         \node [apeer, right =of P51] (P52) {52}; 
         \node [right =of P52] (Prdots) {$\ldots$}; 
         \node [apeer, right =of Prdots] (P100) {100}; 
         \node [apeer, right =of P100] (P101) {101}; 
         \node [apeer, right =of P101] (P102) {102}; 
         
         \node [valL,below =of P1] (V1) {$L$};
         \node [valL,below =of P2] (V2){$L$};
         \node [valL,below =of P3] (V3){$L$};
         \node [below =of Pdots] (Vdots) {$\ldots$}; 
         \node [valL,below =of P51] (V51){$L$};
         \path[line] (P1) -- (V1);
         \path[line] (P2) -- (V2);
         \path[line] (P3) -- (V3);
         \path[line] (P51) -- (V51);
         
         \node [valLp,below =of P52] (V52) {$L'$};
         \node [below =of Prdots] (Vdots) {$\ldots$}; 
         \node [valLp,below =of P100] (V100){$L'$}; 
         \node [valLp,below =of P101] (V101){$L'$};
         \node [valLp,below =of P102] (V102){$L'$};
         \path[line] (P52) -- (V52);
         \path[line] (P100) -- (V100);
         \path[line] (P101) -- (V101);
         \path[line] (P102) -- (V102);

         \begin{pgfonlayer}{background}
            \node[draw=red, minimum height=1.5cm, fit=(P1) (P101)](UNLX) {};
            \node[text=red,above left] at (UNLX.north west) {$X$};
            \node[draw=blue, fit=(P2) (P102)](UNLY) {};
            \node[text=blue,above right] at (UNLY.north east) {$Y$};
         \end{pgfonlayer}
     \end{tikzpicture}
     \caption{Example of stuck network with 99\% UNL overlap and no Byzantine faults.}
     \label{fig:stuckExample}
\end{figure}

\vspace*{1em}

As of the time this paper was written, the recommended XRP trust model has all nodes listening to either a single UNL consisting of $5$ nodes, or a UNL consisting of those $5$ nodes plus one extra node (typically the extra node is oneself; nodes that listen to these extended UNLs are thus called "leaves", since they branch off slightly from the core network). The short-term plan for decentralization involves expanding to a larger, but still agreed-upon, single UNL and diversifying the node operators. Losing forward progress while adjusting to a new node list is not a huge problem (since as soon as everyone agrees on the node list again forward progress will resume, and the previous section guarantees for ``small'' changes it will not fork during the interim); thus we could at least get a positive result by proving that the network cannot get stuck in a complete graph with leaves.

The following lemma simplifies the problem to only needing to verify that complete networks cannot get stuck.

\begin{lemma}\label{addingLeaves}
Suppose $N$ is a closed subset of the network (i.e., the UNL of every node in $N$ is contained in $N$, so that from the perspective of the nodes inside of $N$, $N$ is the entire network) which cannot get stuck and cannot fork. Suppose $\mathcal{P}_i$ is a node not in $N$ such that $\mathsf{UNL}_i=\{\mathcal{P}_i\}\cup N'$, where $N'\subseteq N$ and $|N'|\geqslant q_i$. Then $N\cup\{\mathcal{P}_i\}$ cannot get stuck either.
\end{lemma}
\begin{proof}
Since $N$ cannot get stuck, it is always true that all the nodes in $N$ will eventually fully validate a new ledger. Since $N$ cannot fork, all the nodes in $N$ can only fully validate the same ledger, so eventually there is some ledger $L$ that gets fully validated by every node in $N$. Thus every node in $N$ will validate $L$, and since $|\mathsf{UNL}_i\cap N|=|N'|\geqslant q_i$, $\mathcal{P}_i$ can fully validate $L$ as well. Thus it is always true that $\mathcal{P}_i$ will eventually fully validate a new ledger, so $\mathcal{P}_i$ cannot get stuck.
\end{proof}

\begin{theorem}\label{unstuckTheorem}
Suppose for all nodes the UNL quorum is set to $n-\lfloor (n-1)/k\rfloor$ for some integer $k$. XRP LCP cannot get stuck in a network consisting of a single agreed-upon UNL $X$ of size at least $k$ along with an arbitrary number of leaf validators.
\end{theorem}
\begin{proof}
For any leaf validator $\mathcal{P}_i$, $n_i=|X|+1$; since $q_i=|X|+1-\lfloor (|X|+1-1)/k\rfloor=|X|+1-\lfloor |X|/k\rfloor\leqslant |X|$, by lemma \ref{addingLeaves} it suffices to show that complete networks cannot get stuck.

Thus suppose there is a single agreed-upon UNL $X$. Suppose in some round $r$ all validation messages are delivered quickly enough so that Byzantine accountability holds and every node sees every other node's validations from round $r$. Then the preferred branch algorithm will deterministically push all nodes onto the most popular ledger $L$. Thus in the next round all nodes will validate a child of $L$. If messages sent during deliberation are delivered synchronously, then all nodes see the same proposals from round $1$ and either all exit deliberation with the same ledger (if some transaction set is shared with $80\%$ of the nodes in $X$) or else build their proposal for round $2$ deterministically from the same proposals from round $1$, so all nodes propose the same set of transactions in round $2$, and all nodes leave deliberation with the same set of transactions. Thus every node submits a validation for the same ledger, and all nodes fully validate that ledger.
\end{proof}

\section{Conclusion} \label{conclusion}
We have given a detailed description and thorough analysis of the XRP Ledger Consensus Protocol, which is a protocol for reaching consensus without universal agreement of network participants. Our work corrects prior analysis in \citep{schwartz2014ripple,Armknecht2015}. We show in theorem \ref{safetyTheorem} that roughly $>90\%$ agreement on participants is needed to ensure network safety. In the restricted case of a single expanding UNL with leaves, theorem \ref{unstuckTheorem} shows we can always make forward progress during periods when no nodes are faulty and network messages are delivered with bounded delay. In the more general case with even minor disagreement of participants, we cannot guarantee that the network makes forward progress.

It is an open question whether the sufficient overlap condition in theorem \ref{safetyTheorem} can be improved by a more detailed consideration of the trust topology of the network. A more complicated condition that does not simply take into account pairwise overlaps but also the way in which messages flow indirectly through the network might have potential for giving a more precise condition for guaranteeing safety. Likewise, we might be able to leverage the trust structure to better explain cases when deliberation can fail, which in turn might allow a more refined understanding on forward progress.

Although we have shown that XRP LCP is provably safe with the current and near-future network structures, in an attempt to alleviate some of its shortcomings, in the sibling paper \citep{Ethan} we present an alternative consensus protocol called Cobalt. Similar to XRP LCP, Cobalt can also be used in networks that lack uniform agreement on participants or trust, but makes forward progress at a steady rate in the presence of maximum tolerated Byzantine faults and arbitrary asynchrony. It only needs $>60\%$ overlap to match the XRP LCP safety tolerances. Cobalt also has several other properties that make it simpler to analyze the health of networks in practice. For these reasons we believe Cobalt represents an encouraging direction for even greater future decentralization of the XRP network.

\paragraph{Acknowledgments.} Thank you to Haoxing Du and Joseph McGee for collaborating on the early stages of this research. We also thank David Schwartz and Stefan Thomas for providing useful discussions and guidance, and Rome Reginelli for careful editing. We lastly thank all our colleagues at Ripple for their support. This work was funded by Ripple.

\bibliography{refs.bib}

\appendix
\section{Algorithms} \label{appendix}
In this appendix, we provide pseudo-code for the three components of Ripple consensus described in section~\ref{protocol}. Note that under the network model in section~\ref{definitions}, a message which is \textbf{broadcast} by a node $\mathcal{P}_i$ is \textbf{receive}d by all nodes, including $\mathcal{P}_i$ itself.
\begin{algorithm}
    \begin{algorithmic}[1]
        \State $s_{max} \gets 0$ \Comment{Track the largest validated ledger sequence number }
        \item[]
        \Function{Start}{$L$}
            \State $\tilde{L} \gets L, r \gets 0$
            \State $T \gets $ pending transactions
            \State $props \gets \{\}$ \Comment{$props$ is a map from node to proposal}       
            \State Initialize $props$ with previously received proposals for $\tilde{L}$ 
            \State \textbf{broadcast} $P_{T,r,\tilde{L}, i}$
        \EndFunction
        \item[]
        \Receive{$P_{T', r', L, j}$}
            \If{$\mathcal{P}_j \in \mathsf{UNL}_i$ and $\tilde{L}=L$ and $r' > props[j].r$}
                \State $props[j] = P_{T', r', L, j}$
            \EndIf
        \EndReceive
        \item[]
        \Function{Update}{$ $} \Comment{Called at a regular, protocol defined interval}
            \If{$\tilde{L} \neq$\Call{PreferredLedger}{$ $}}
                \State \Call{start}{\textproc{preferredLedger}()}
            \Else
                \State \Call{UpdatePosition}{$ $}
                \If{\Call{CheckConsensus}{$ $}}
                    \State $\tilde{L} \gets \Call{apply}{T,\tilde{L}}$
                    \If{$\mathbf{seq}(\tilde{L}) > s_{max}$}
                        \State \textbf{broadcast} $V_{\tilde{L}, i}$
                        \State $s_{max} \gets \mathbf{seq}(\tilde{L})$
                    \EndIf
                    \State \Call{Start}{$\tilde{L}$}
                \EndIf
            \EndIf
        \EndFunction
        \algstore{deliberation}
    \end{algorithmic}
    \caption{Deliberation from the perspective of $\mathcal{P}_i$}
    \label{alg:deliberation}
\end{algorithm}
\clearpage
\begin{algorithm}
    \ContinuedFloat
    \begin{algorithmic}
        \algrestore{deliberation}
        \Function{UpdatePosition}{$ $}
              \State $T_{all} \gets \bigcup_{P \in props} P.T$ \Comment{Set of all proposed transactions}
              \State $\tau \gets threshold(r)n_i$  
              \State $T \gets \{ x \in T_{all} : \Call{support}{x} > \tau\}$ \Comment{\parbox[t]{.41\linewidth}{\textproc{support} is number of nodes proposing $x$} }
              \State $r \gets r + 1$
              \State \textbf{broadcast} $P_{T,r,\tilde{L}, i}$
          \EndFunction
        \item[] 
        \Function{CheckConsensus}{$ $}
            \State $n_a \gets |\{P \in props : P.T = T\}|$ \Comment{\parbox[t]{0.5\linewidth}{Node positions agreeing with our position}}
            \State \Return{$n_a \geq q_i$} 
        \EndFunction
    \end{algorithmic}
    \caption{Deliberation from the perspective of $\mathcal{P}_i$ (continued) }
\end{algorithm}

\begin{algorithm}
    \begin{algorithmic}[1]
        \State $vals = \{\}$ \Comment{$vals$ is a map from $L$ to the set of nodes that validated $L$}
        \item[]
        \Receive{$V_{L, j}$}
            \If{$\mathcal{P}_j \in \mathsf{UNL}_i$}
                \State $vals[L] \gets vals[L] \cup j$
                \If{$|vals[L]| \geq q_i$ and $\mathsf{seq}(L) > \mathsf{seq}(\hat{L})$ }
                    \State $\hat{L} \gets L$
                \EndIf
            \EndIf
        \EndReceive
    \end{algorithmic}
    \caption{Validation from the perspective of $\mathcal{P}_i$}
    \label{alg:fullvalidation}
\end{algorithm}

\begin{algorithm}
    \begin{algorithmic}[1]
        \State $lastVals = \{\}$ \Comment{\parbox[t]{0.7\linewidth}{$lastVals$ is a map from trusted node to its most recent validated ledger }}
        \item[]
        \Receive{$V_{L, j}$}
            \If{$\mathcal{P}_j \in \mathsf{UNL}_i$} 
                \State $lastVals[j] \gets L$
            \EndIf
        \EndReceive
        \item[]

        \Function{PreferredLedger}{$ $ }  
            \State $L \gets $ earliest common ancestor of ledgers in $lastVals$
            \State $done = $ False
            \While{$|children(L)| > 0$ and not $done$}
                \State $C \gets  $ \parbox[t]{0.8\linewidth}{Sorted array of $children(L)$ by descreasing $\mathsf{supp}_{branch}$, breaking ties with $\phi$}      
                \State $\Delta \gets \mathsf{supp}_{branch}{C[0]}$
                \If{$|children(L)| > 1$}
                    \State $\Delta \gets \Delta - \mathsf{supp}_{branch}{C[1]} + \phi(C[0],C[1])$
                \EndIf
                \If{$\Delta > \mathsf{uncommitted}(\mathsf{seq}(L)+1) $}
                    \State $L \gets C[0]$
                \Else
                    \State $done \gets $ True
                \EndIf
            \EndWhile
            \If{$L \in ancestors(\tilde{L})$}
                \State \Return{$\tilde{L}$}
            \Else
                \State \Return{$L$}
            \EndIf
        \EndFunction
    \end{algorithmic}
    \caption{Preferred branch from the perspective of $\mathcal{P}_i$}
    \label{alg:preferredBranch}
\end{algorithm}

\end{document}